\documentclass[a4paper,USenglish,cleveref,autoref]{lipics-v2021}

\usepackage{microtype}

\title{A Practical Dynamic Programming Approach to Datalog Provenance Computation}

\author{Yann Ramusat}{DI ENS, ENS, CNRS, PSL University \& Inria, France}{yann.ramusat@ens.fr}{https://orcid.org/0000-0002-7181-8331}{}

\author{Silviu Maniu}{Université Paris-Saclay, LISN, CNRS, France}{silviu.maniu@lisn.upsaclay.fr}{https://orcid.org/0000-0002-8623-1533}{}

\author{Pierre Senellart}{DI ENS, ENS, CNRS, PSL University \& Inria \& IUF, France}{pierre@senellart.com}{https://orcid.org/0000-0002-7909-5369}{}

\authorrunning{Y.\, Ramusat, S.\, Maniu, and P.\, Senellart}

\Copyright{Yann Ramusat, Silviu Maniu, and Pierre Senellart}

\ccsdesc[500]{Information systems~Data provenance}
\ccsdesc[500]{Theory of computation~Data provenance}

\keywords{Datalog, provenance, semiring, graph databases, dynamic programming}

\funding{This work has been funded by the French government under
management of Agence Nationale de la Recherche as part of the
“Investissements d’avenir” program, reference ANR-19-P3IA-0001
(PRAIRIE 3IA Institute).}

\usepackage{apxproof}
\newtheoremrep{theorem}{Theorem}
\newtheoremrep{lemma}{Lemma}

\renewcommand{\leq}{\leqslant}
\renewcommand{\geq}{\geqslant}
\renewcommand{\epsilon}{\varepsilon}
\renewcommand{\phi}{\varphi}

\usepackage{algorithm}
\usepackage[noend]{algpseudocode}

\usepackage[fleqn]{amsmath}

\let\algo\textsc

\usepackage{pgfplots}
\usepgfplotslibrary{fillbetween}
\pgfplotsset{compat=1.15}
\pgfplotsset{
  y tick label style={/pgf/number format/.cd, set thousands separator={\,}}
}

\usepackage{mathtools}

\renewcommand{\inf}{\mathop{\mathsf{inf}}\limits}

\newcommand{\prov}{\mathop{\mathsf{prov}}}

\nolinenumbers

\begin{document}

\maketitle

\begin{abstract}
We establish a translation between a formalism for dynamic programming
  over hypergraphs and the computation of semiring-based provenance for Datalog programs.
The benefit of this translation is a new method for computing provenance for a specific class of semirings. 
Theoretical and practical optimizations lead to an efficient implementation using \textsc{Soufflé}, a state-of-the-art Datalog interpreter. 
Experimental results on real-world data suggest this approach to be efficient in practical contexts, even competing with our previous dedicated solutions for computing provenance in annotated graph databases. 
The cost overhead compared to plain Datalog evaluation is fairly moderate
  in many cases of interest.
\end{abstract}

\section{Introduction}
\label{sec:introduction}

A notion of \textit{provenance} for Datalog queries was introduced by
Green, Karvounarakis, and Tannen~\cite{green_provenance_2007}. It is
based on the algebraic structure of \textit{semirings} to encode
additional meta-information about query results and extends the notion of 
semiring provenance of the positive fragment of the relational algebra, also
introduced in~\cite{green_provenance_2007}. The \textit{full provenance}
of a Datalog program (i.e., the provenance associated to each
\textit{derived tuple}) is expressed as a \textit{system of equations}
over an \emph{$\omega$-continuous} semiring.

Beyond this initial definition, some research has proposed other
representation frameworks to improve the computation of Datalog
provenance. Circuits~\cite{deutch_circuits_2014} have been proposed to
obtain a compact representation of the provenance. Derivation-tree
analysis~\cite{esparza_solving_2011} focuses on the natural mapping
between systems of equations over $\omega$-continuous semirings and
context-free grammars. These techniques rely on additional properties of
the semiring, thus being limited to some specific applications; in
particular, some semiring properties such as \emph{absorptivity} serve as
a basis for optimization in both scenarios.

The present work addresses two major challenges related to the
computation of semiring-based provenance of Datalog programs.

On the first hand, much of the knowledge around the theory and
applications of using semirings to capture some meta-information about a
computation (i.e., \emph{provenance} in database speak) is spread across
different fields of computer sciences. Among other areas, semirings have
been successfully used in linguistic structure
prediction~\cite{smith_linguistic_2011}, dynamic
programming~\cite{huang_advanced_2008}, constraint
programming~\cite{bistarelli_semiring-based_1997}, and variations of the
shortest-distance problem in graphs~\cite{Mohri:2002:SFA:639508.639512}.
The structure of semirings is also at the heart of the study of formal
languages~\cite{rozenberg_handbook_1997, Kuich:1985:SAL:576931}, weighted
automata~\cite{droste_handbook_2009}, graph theory and graph
algorithms~\cite{minoux:hal-01304880}, and computational
algebra~\cite{parallel}. The \textit{algebraic path
problem}~\cite{general} provides a common viewpoint for
a large variety of problems involving semirings.
Because of the widespread use of semirings in different areas, and rather
unfortunately, it is common to find 
different names used for similar properties (e.g., $0$-closed
semirings are also called absorptive, bounded, simple, $0$-stable,
etc.). Even more concerning, the term ``semiring'' itself is replaced
by similar names in some works (e.g. by ``dioids'' or ``algèbre de
chemin''~\cite{RO_1975__9_1_77_0}); in some cases, the semiring
structure could have been used to model the problem, but was not
acknowledged at the time~\cite{knuth_generalization_1977}. We believe it
is necessary to spend some effort establishing 
links between multiple representation frameworks of what is essentially
semiring provenance to
benefit from a wide range of techniques proposed in various domains of
computer science. This would
obviate the pitfall of designing algorithms and developing techniques in
the jargon of a given field, whereas some key ideas may have already been
investigated under similar algebraic frameworks, but by other research
communities.

On the other hand, to the best of the knowledge of the authors, the
bulk of the scientific literature about provenance for Datalog programs
is mostly concerned with, either, formalization and
extensions~\cite{DBLP:journals/corr/abs-2105-14435} to the theoretical
framework behind semiring-based provenance, with the objective to reach a
massive number of application domains, or elegantly leveraging
quite abstract tools such as circuits~\cite{deutch_circuits_2014} or
Newton convergence~\cite{esparza_solving_2011} to optimize computations.
Those works rarely go all the way to the implementation level,
apart from some 
``proofs of concept''~\cite{10.1007/978-3-319-08846-4_1}, raising the
question of the applicability to real-world data.

A work of note that does come with an implementation
is~\cite{7113407,Deutch2018EfficientPT}, concerned about the computation
of a \textit{relevant} subset of provenance information through selection
criteria based on tree patterns and the ranking of rules and facts, when
the computation of the \textit{full} provenance (how-provenance) is
unfeasible in practice.

Our contributions can be organized into four parts.
\begin{enumerate}
	\item We first establish a correspondence between dynamic programming over hypergraphs (as introduced in~\cite{huang_advanced_2008} under the name of \textit{weighted hypergraphs}) and the \textit{proof-theoretic} definition of the provenance for Datalog programs. 
		We provide both-way translations and characterize for
                which class of semirings the \textit{best-weight
                derivation} in the hypergraph corresponds to the
                provenance of the initial Datalog program.
	\item The translation we thus introduced permits us to obtain a
          version of Knuth's generalization of Dijkstra's
          algorithm~\cite{knuth_generalization_1977} to the \emph{grammar
          problem}, adapted to the case
          of Datalog provenance computation. 
		In the special setting where all hyperedges are of
                arity~1, we obtain the classical notion of semiring-based
                provenance for graph
                databases~\cite{ramusat:hal-01850510}. In the general
                setting, the algorithm steadily generalizes to Datalog the adapted Dijkstra's algorithm we presented in~\cite{ramusat:hal-01850510}, under the same assumptions on the properties of the underlying semiring.
	\item Such algorithm is unlikely to be efficient as-is in practical contexts. 
		The main issue is closely related to the inefficiency of basic Datalog evaluation: many computations of facts (provenance values) have already been deduced, leading to redundant computations. 
		We show that the \textit{semi-naïve} evaluation strategy
                for Datalog is also applicable in our setting. An added
                advantage is that it 
		greatly facilitates the extension of existing Datalog solvers to compute provenance annotations.
		We implement our strategy extending \textsc{Soufflé}~\cite{10.1007/978-3-319-41540-6_23}, a state-of-the-art Datalog solver. 
		We apply our solution to process rich graph queries
                (translated into Datalog programs) on several real-world
                and synthetic graph datasets, as well as to Datalog
                programs used in previous
                works~\cite{Deutch2018EfficientPT}.
		The performance of the implementation competes with
                dedicated solutions specific to graph databases~\cite{ramusat:hal-03140067}.
	\item The link we established with the framework of weighted
          hypergraphs, and hence, with the grammar problem, opens up new opportunities and challenges for semiring-based provenance for Datalog programs. 
		We thus carefully discuss selected topics and problems
                that can now understood as closely linked to the provenance framework~\cite{green_provenance_2007}.
\end{enumerate}

We start by introducing in
Section~\ref{sec:background} basic concepts on semirings and we recall the definition of provenance for Datalog programs.
We formulate and prove in Section~\ref{sec:correspondence} the correspondence between weighted hypergraphs and semiring-based provenance for Datalog programs.
In Section~\ref{sec:algorithm}, we present the adapted version of Knuth's
algorithm for the grammar problem and discuss theoretical aspects of its optimization. 
We then dive into the practical aspects of its implementation using \textsc{Soufflé}, and present in Section~\ref{sec:implementation} experimental results witnessing the efficiency of our approach for practical scenarios. 
We finally discuss, in Section~\ref{sec:opportunities} and
Section~\ref{sec:related-work}, respectively opportunities raised by our
translation and related work.
For space reasons, proofs and some auxiliary material are deported to an
appendix.

\section{Background}
\label{sec:background}


In the following, we recall basic concepts of semiring theory
underlying the optimization techniques we provide in this paper. For more
background on the theory and applications of semirings, examples of
relevant semirings, as well as references to the literature on advanced
notions of semiring theory, see our previous
work~\cite{ramusat:hal-03140067}. We follow the definitions
in~\cite{ramusat:hal-03140067} and highlight notions that
occur under different names depending on the
application domain.

\begin{definition}[Semiring]
A \emph{semiring} is an algebraic structure
$(S,\oplus,\otimes,\bar{0},\bar{1})$ where $S$ is some
  set, $\oplus$ and $\otimes$ are binary operators over~$S$, and
  $\bar{0}$ and $\bar{1}$ are elements of $S$, satisfying the
  following axioms:
\begin{itemize}
  \item $(S,\oplus,\bar{0})$ is a \emph{commutative monoid}: $(a\oplus b)\oplus c =
    a\oplus(b\oplus c)$, $a\oplus b=b\oplus a$, $a\oplus
    \bar{0}=\bar{0}\oplus a=a$;
  \item $(S,\otimes,\bar{1})$ is a \emph{monoid}: $(a\otimes b)\otimes
    c=a\otimes(b\otimes c)$, $\bar{1}\otimes a=a\otimes \bar{1}=a$;
  \item $\otimes$ distributes over $\oplus$: $a\otimes(b\oplus c)=(a\otimes
  b)\oplus(a\otimes c)$; \item $\bar{0}$ is an annihilator for $\otimes$: $\bar{0}\otimes a=a\otimes \bar{0} = \bar{0}$.
\end{itemize}
\end{definition}

A semiring is \emph{commutative} if for all $a,b\in S$, $a
\otimes b=b \otimes a$.
A semiring is \emph{idempotent}
if for all $a\in S$, $a\oplus a=a$. For an idempotent semiring
we can introduce the \emph{natural order}
defined by $a \leq b$ iff $a \oplus b = a$.%
\footnote{There are
unfortunately two definitions of natural order commonly found
in the literature; we use here that found
in~\cite{Mohri:2002:SFA:639508.639512,huang_advanced_2008} which 
matches the standard order on the tropical semiring; other
works~\cite{rozenberg_handbook_1997,green_provenance_2007,ramusat:hal-03140067} define it as
the reverse order. Our choice obviously has some impacts: in particular, when
defining Datalog provenance, we need greatest fixpoints in lieu of the least
fixpoints used in~\cite{rozenberg_handbook_1997,green_provenance_2007}.}
Note that this order is compatible with the two binary operations of the
semiring: for all $a, b, c \in S$, $a \leq b$ implies $a \oplus c \leq b
\oplus c$ and $a \otimes c  \leq b \otimes c $. This is also called the
\emph{monotonicity} property.

An important property is that of
\emph{k-closedness}~\cite{Mohri:2002:SFA:639508.639512}, i.e., a semiring is $k$-closed if:
\(
  \forall a\in S, ~\bigoplus_{i=0}^{k+1} a^i =
  \bigoplus_{i=0}^k a^i.\)
Here, by $a^i$ we denote the repeated application of the $\otimes$ operation $i$
times, i.e., $a^i=\underbrace{a\otimes a \otimes \dots\otimes a}_i$.
0-closed semirings (i.e., those in which $\forall a\in S,
\bar{1}\oplus a=\bar{1}$)
have also been called \emph{absorptive}, \emph{bounded}, or \emph{simple}
depending on the literature. Note that any $0$-closed semiring is
idempotent (indeed, $a\oplus a=a\otimes(\bar{1}\oplus
\bar{1})=a\otimes \bar{1}=a$) and therefore admits a natural order.

Huang~\cite{huang_advanced_2008} introduces the notion of
\emph{superiority} of a semiring $S$ with respect to a partial
order~$\leq$, defined by:
$\forall a,b\in S\: a\leq a\otimes b, \: b\leq a\otimes b$.
The
natural order satisfies this notion for $0$-closed semirings:

\begin{lemmarep}\label{lem:superiority}
  Let $S$ be an idempotent semiring and $\leq$ the natural order over
  $S$. Then $S$ is \emph{superior} with respect to $\leq$
  if and only if $S$ is $0$-closed.
\end{lemmarep}

\begin{proof}
  First assume $S$ superior with respect to~$\leq$. Then for any $a$,
  $\bar 1\leq \bar 1\otimes a=a$,
  which means that $\bar 1+a = \bar 1$, i.e., $S$ is $0$-closed.

  Now assume $S$ $0$-closed.
Since $a \oplus a \! \otimes \! b = a \otimes (\bar{1} \oplus b) = a$, we
  have: $a \leq a \otimes b$, and similarly for $b\leq a\otimes b$. Thus
  $S$ is superior with respect to~$\leq$.
\end{proof}

An easier way of understanding natural order in 0-closed semirings 
is to note that for
any idempotent semiring $S$, $\bar{0}$ is the greatest element ($\forall
a \in S$, $a\oplus 0=a \leq \bar{0}$) while, if the semiring is also 0-closed (i.e.,
\emph{bounded}), $\bar{1}$ is the smallest ($\forall a \in S$, 
$\bar{1}\oplus a=\bar{1} \geq a$).
Thus a bounded semiring~$S$ verifies $\bar{1} \leq a \leq \bar{0}$ for all
$a\in S$.


\begin{definition}[$\omega$-Continuous semiring] 
  A semiring $(S,\oplus,\otimes,\bar 0,\bar 1)$ is \emph{$\omega$-continuous} if
\begin{enumerate}
  \item $(S, \geq)$ is a \emph{$\omega$-complete} partial order, i.e.,
    the infimum $\inf_{i \in
    \mathbb{N}} \: a_i$ of any infinite chain $a_0 \geq a_1 \geq \dots $
    exists in $S$.
  \item both addition and multiplication are \emph{$\omega$-continuous} in both
  arguments, i.e., for all $a\in S$ and infinite chain  $a_0 \geq a_1
    \geq \dots $, $(a\oplus\inf_{i\in\mathbb{N}} a_i)=\inf_{i\in\mathbb{N}}
    (a\oplus a_i)$, $(a\otimes \inf_{i\in\mathbb{N}} a_i)=\inf_{i\in\mathbb{N}}
    (a\otimes a_i)$, $(\inf_{i\in\mathbb{N}} a_i)\otimes a=\inf_{i\in\mathbb{N}}
    (a_i\otimes a)$.
\end{enumerate}  
\end{definition}

In such semirings we can define countable sums: $\bigoplus\limits_{n \in
\mathbb{N}} a_n = \underset{m \in \mathbb{N}}{\mathsf{inf}} \bigoplus\limits_{i = 0}^{m} a_i$.

A \emph{system of fixpoint equations} over an $\omega$-continuous
semiring $S$ is a finite set of equations:
\(
		X_1  =  f_1(X_1,X_2, \dots, X_n) ,\dots,
		X_n  =  f_n(X_1,X_2, \dots, X_n),
                \)
where $X_1,\dots,X_n$ are variables and $f_1,\dots f_n$ are polynomials with coefficients in~$S$.
We extend the notion of natural order from semiring elements to
tuples of semiring elements by simply considering the product order.
We then have the following on solutions of a system
of equations over an $\omega$-continuous semiring:

\begin{theorem}[Theorem 3.1 of~\cite{rozenberg_handbook_1997}]
  Every system of fixpoint equations $\mathbf{X} = f(\mathbf{X})$ over
  a commutative $\omega$-continuous semiring has a greatest solution
  $\mathsf{gfp}(f)$ w.r.t. $\leq$, and $\mathsf{gfp}(f)$ is equal to the
  infimum of the
Kleene sequence:
  \( \mathsf{gfp}(f) = \inf_{m \in \mathbb{N}} f^m (\mathbf{\bar 0}).\)
\end{theorem}


We now recall some basics about the Datalog query language and refer to
\cite{DBLP:books/aw/AbiteboulHV95} for more details. A \textit{Datalog
rule} is of the form $R(\vec{x}) \; \textsf{:-} \; R_1(\vec{x_1}), \:
\dots, \: R_n(\vec{x_n})$ with $R$'s representing relations of a given arity
and the $\vec{x}$'s
tuples of variables of corresponding arities. Variables occurring on the
left-hand side, the \textit{head} of the rule, are required to occur in
at least one of the atoms on the right-hand side, the \textit{body} of
the rule. A \textit{Datalog program} is a finite set of Datalog rules. We
call \textit{fact} a rule with an empty body and variables replaced by
constants. We divide relations into \emph{extensional} ones (which can
only occur as head of a fact, or in rule bodies) and \emph{intensional}
ones (which may occur as heads of a non-fact rule). The set of
extensional facts is called the \emph{extensional database}
($\mathsf{EDB}$). We distinguish one particular relation occurring in the
head of a rule, the \emph{output predicate} of the
Datalog program. We refer to \cite{DBLP:books/aw/AbiteboulHV95} for the
semantics of such a program and the notion of \emph{derivation tree}.

There are two ways of defining the \emph{provenance} of a Datalog program
$q$ with output predicate $G$ over an $\omega$-continuous semiring. We
can first base this definition on the proof-theoretic definition of standard Datalog:

\begin{definition}[Proof-theoretic definition for Datalog provenance \cite{green_provenance_2007}]
\label{def:ptdfdp}
Let $(S, \oplus, \otimes, \bar{0}, \bar{1})$ be a commutative
  $\omega$-continuous semiring and $q$ a Datalog program with output
  predicate $S$ and such that all extensional facts $R(t')$ are annotated with an
  element of $S$, denoted as $\prov_R^q(t')$. Then the \emph{provenance} of
  $G(t)$ for $q$, where $G(t)$ is in the output of $q$, is defined as:
  \[
    \mathsf{prov}_G^q (t) = \bigoplus\limits_{\tau \: \textup{yields} \:
    t} \left( \bigotimes\limits_{t' \in \: \textup{leaves}(\tau)}
    \mathsf{prov}_R^q(t') \right).\]
\end{definition}

The first sum ranges over all the derivation trees of the fact $t$ (see
Figure~\ref{fig:ex-dr-datalog} for examples of derivation trees), the
second sum ranges over all leaves of the tree (extensional facts). This
definition describes how the provenance propagates across the deduction
process given an initial assignment of provenance weights to the
extensional relations of $q$, $\mathsf{prov}_R^q$.

\begin{figure}
	\begin{tikzpicture}[sibling distance = 2.5cm, level distance = 1cm, scale=0.9,every node/.style={scale=0.9}]
		\node {\sffamily path(Paris, London) 3}
			child {node {\sffamily edge(Paris, London) 3} 
				child [white] {node {dummy}} edge from
                                parent node [right] {$r_1$} };
	\end{tikzpicture}
	\begin{tikzpicture}[sibling distance = 3.5cm, level distance = 1cm, scale=0.9,every node/.style={scale=0.9}]
		\node {\sffamily path(Paris, London) 1}
			child	{node {\sffamily path(Paris, Lille) 1}
					child { node {\sffamily
                                        edge(Paris, Lille) 1} edge from
                                        parent node [right] {$r_1$} }}
			child {node {\sffamily edge(Lille, London) 0}
                        edge from parent node [right, xshift = -1.15cm]
                        {$r_2$} };
	\end{tikzpicture}
	\parbox[b]{2cm}{ \scriptsize 
					\begin{alignat*}{4}
					\multispan3{\hfil
                                          \textsf{edge(Paris, London) :-} }										&& 3			\\
					\multispan3{\hfil
                                          \textsf{edge(Paris, Lille) :-} }										&& 1			\\
					\multispan3{\hfil
                                          \textsf{edge(Lille, London) :-} }	 									&& 0			\\
					\textsf{path(x, y)}
                                          \; \;	&	\textsf{:-}	&
                                          \; \textsf{edge(x, y)}
                                          && \quad r_1	\\
					\textsf{path(x, y)}				\; \; 	&	\textsf{:-}	& \; \textsf{path(x, z),}	&&			\\
															&				& \; \textsf{edge(z, y)}	&& 
\quad
r_2		
\end{alignat*}
					}
	\caption{Derivation trees along their weights for the fact $\textsf{path(Paris, London)}$ 
	using the transitive closure Datalog program over the tropical semiring
	with an $\mathsf{EDB}$ containing $3$ facts}
	\label{fig:ex-dr-datalog}
\end{figure}
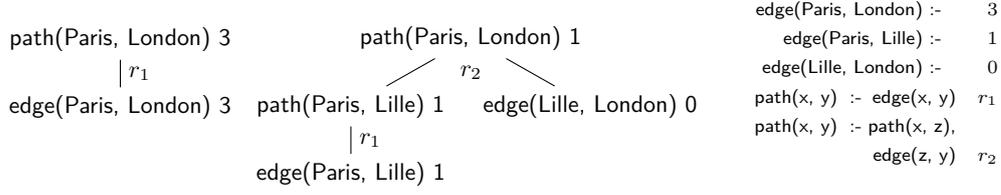

\begin{example}
  The \emph{tropical semiring} is $(\mathbb{R}^+\cup\{\infty\}, \min, +,
  \infty, 0)$. We show in Figure~\ref{fig:ex-dr-datalog} an example
  Datalog program (right) with tropical semiring annotations on
  extensional facts, as well as the (only) two derivation trees of the
  fact $\textsf{path(Paris, London)}$ along their weight. This witnesses
  that the provenance of $\textsf{path(Paris, London)}$ is $\min(1,3)=1$.
\end{example}

Since some
tuples can have infinitely many derivations, the Datalog semantics given
above cannot be used as an algorithm.
As pointed out in~\cite{green_provenance_2007} it is possible instead to
use a fixpoint-theoretic definition of the provenance of a Datalog
query~$q$: introduce a
fresh variable for every possible \emph{intensional} tuple (i.e., every
possible \emph{ground atom}), and produce
for this variable an equation that reflects the
\emph{immediate consequence operator} $T_q$ -- \emph{extensional} facts
appearing as their semiring annotations in these equations. This yields a
system of fixpoint equation $f_q$. The
provenance of $G(t)$ for $q$ is now simply the value of the variable
corresponding to $G(t)$ in $\mathsf{gfp}(f_q)$.

The fixpoint-theoretic definition directly yields an algorithm, albeit a very
inefficient one because of the need of generating a rule for every
intensional tuple. In this paper, we investigate more efficient
algorithms, for specific types of semirings.

\section{Datalog provenance and dynamic programming over hypergraphs}
\label{sec:correspondence}

We now show how to convert a Datalog program into a weighted hypergraph
(as introduced in \cite{huang_advanced_2008}) and characterize the
semirings where the best-weight derivation in the hypergraph corresponds
to the provenance for the initial Datalog program, mimicking the
proof-theoretic definition. We first
recall basic definitions and notation related to hypergraphs.


\begin{definition}[Weighted hypergraph~\cite{huang_advanced_2008}]
  Given a semiring $S$, a \emph{weighted hypergraph} on~$S$ is a pair $H
  = \langle  V, E \rangle$, where $V$ is a finite set of vertices and
  $E$ is a finite set of hyperedges, where each element $e\in E$ is 
        a triple $e = \langle h(e), \mathrm{T}(e), f_e \rangle$ with
        $h(e) \in V$ its \emph{head} vertex, $\mathrm{T}(e) \in V$ an
        ordered list of \emph{tail} vertices and $f_e$ a \emph{weight
        function} from $S^{|\mathrm{T}(e)|}$ to $S$.
\end{definition}

We note $|e| = |\mathrm{T}(e)|$ the \textit{arity} of a hyperedge. If $|e| = 0$,
we say $e$ is nullary and then $f_e()$ is a constant, an element of the
semiring; we assume there exists at most one nullary edge for a given
vertex.
In that case, $v = h(e)$ is called a \emph{source vertex} and we note
$f_e()$ as $f_v$. The \textit{arity of a hypergraph} is the maximum
arity of any hyperedge.

The \textit{backward-star} $\mathrm{BS}(v)$ of a vertex $v$ is the set of incoming hyperedges  $ \{ e \in E \mid h(e) = v \} $.
The \emph{graph projection} of a hypergraph $H = \langle V, E \rangle$ is a directed graph $G=(V, E')$
where $E' = \{ (u,v) \mid \exists e \in \mathrm{BS}(v),  u \in \mathrm{T}(e) \}$.
A hypergraph $H$ is acyclic if its graph projection $G$ is acyclic; 
then a topological ordering of $H$ is an ordering of $V$ that is a
topological ordering of~$G$.

With these definitions in place, we can encode a Datalog program with
semiring annotations as a weighted hypergraph in a straightforward
manner:
\begin{definition}[Hypergraph representation of a Datalog program]
	Given a Datalog program $q$ as a set of rules $\{q_1, \cdots,
        q_n\}$ and the semiring $S$ used for annotations, we define the \emph{weighted hypergraph representation of $q$} as $H_q = \langle V_q,E_q \rangle$ with $V_q$ being 
	all ground atoms and, for each instantiation of a rule
        $t(\vec{x}) \leftarrow r_1(\vec{x_1}), \dots , r_n(\vec{x_n})$,
         a corresponding  edge 
$\langle t(\vec{x}) , ( r_1(\vec{x_1}), \dots , r_n(\vec{x_n}) ) , \otimes \rangle $. For a fact $R(\vec{x}) \in \mathsf{EDB}(q)$ we add a nullary edge $e$ with $h(e) = R(\vec{x})$ and $f_e = \mathsf{prov}_R^q(\vec{x})$.
\end{definition}



The notion of \textit{derivations} is the hypergraph counterpart to
paths in graph. We recall the definition of derivations and we define it
in a way that is reminiscent of Datalog-related notions.

\begin{definition}[Derivation in hypergraph~\cite{huang_advanced_2008}]
We recursively define a \emph{derivation} $D$ of a vertex $v$ in a
  hypergraph $H$ (as a pair formed of a hyperedge and a list of
  derivations), its size $|D|$ (a natural integer) and its weight $w(D)$
  (a semiring element) as follows:
\begin{itemize}
\item If $e \in \mathrm{BS}(v)$ with $|e| = 0$, then $D = \langle e,
  \langle\rangle \rangle $ is a derivation of $v$, $|D| = 1$, and $w(D) = f_e()$. 
\item If $e \in \mathrm{BS}(v)$ with $|e| \geq 0$, $D_i$ is a derivation of
  $T_i(e)$ for $i=1\dots |e|$, then
$D =  \langle e, \langle D_1 \cdots D_{|e|} \rangle \rangle$ is a derivation of $v$, 
$|D| = 1 + \sum_{i = 1}^{|e|} |D_i|$, $w(D) = f_e(w(D_1), \dots ,w(D_{|e|}))$.
\end{itemize}

  We note $\mathcal{D}_{H}(v)$ the set of derivations of $v$ in $H$. 
\end{definition}


When modeling Datalog provenance in a semiring~$S$ as weighted
hypergraphs on~$S$, all non-source weight functions are bound to the $\otimes$
operation of the semiring. Note that, if $S$ is idempotent, the natural
order on $S$ induces an ordering on derivations: $D \leq D'$ if $w(D)
\leq w(D')$.


We now show that in this formalism, the Datalog provenance of an output predicate can be understood as the best-weight for the corresponding vertex in the hypergraph. 

\begin{definition}[Best-weight~\cite{huang_advanced_2008}]
	The \emph{best-weight} $\delta_H(v)$ of a vertex $v$ of a
        hypergraph $H$ on a semiring~$(S,\oplus,\otimes,\bar 0,\bar 1)$ is the weight of the best derivation of $v$:
        \[\delta_H(v) = \left\{
    \begin{array}{ll}
        f_v & \mbox{if v is a source vertex;} \\
        \bigoplus_{D \in \mathcal{D}_{H}(v)} w(D) & \mbox{otherwise.}
    \end{array}
  \right.\]
\end{definition}

The best-weight generally requires additional properties of either the
hypergraph or the semiring to be well-defined. Acyclicity for the
hypergraph is a sufficient condition. Existence of an infinitary
summation operator in the semiring extending $\oplus$, guaranteed in
$\omega$-continuous semirings, is also a
sufficient condition.
To guarantee semantics compatible with the
intuitive meaning of provenance, a more restrictive sufficient condition
is for the semiring to be a \textit{$c$-complete
star-semiring}~\cite{krob_monoides_1987}, see \cite{ramusat:hal-03140067}
for details.

We can now show that Datalog provenance can be computed through the
formalism of weighted hypergraphs.
Let us start with a lemma exhibiting a one-to-one mapping
between derivations in the hypergraph and proofs in Datalog. 
\begin{lemmarep} \label{lemma:mult}
  For any Datalog query~$q$ and grounding of an atom $v$ of $q$, there is
  a bijection between $\mathcal{D}_{H{q}}(v)$ and $\{\tau \mid \tau \mbox{ yields } v\}$.
\end{lemmarep}

\begin{proof} 
	By definition of $H_q$ each instantiation of a rule corresponds
        to a unique hyperedge. Then, we can inductively construct for a
        given derivation $D$ its associated (unique) Datalog proof tree $\tau_D$:
	\begin{itemize}
		\item If $|D| = 1$, then $v$ is a source vertex and thus an extensional tuple, we get the empty proof.
		\item If $|D| \geq 1$, then there exists $e \in \mathrm{BS}(v)$ where
                  $|e| \geq 0$ and $D_i$ a derivation of $T_i(e)$ for $1
                  \leq i \leq |e|$, where $D =  \langle e, D_1 \cdots D_{|e|} \rangle$. 
		By definition, this hyperedge corresponds to the grounding of a rule $t(\vec{x}) \leftarrow r_1(\vec{x_1}), \dots , r_n(\vec{x_n})$. 
		By induction, for $1 \leq i \leq |e|$, $\tau_{D_i}$ is the corresponding proof of the derivation $D_i$. Then by composition we obtain $\tau_D$ the proof for $D$.
	\end{itemize}
\end{proof}

We then show that the weight of each derivation of a tuple is equal to
the corresponding proof tree weight in Datalog.

\begin{lemmarep} \label{lemma:sum}
  For any Datalog query $q$ and grounding of an atom $v$ of $q$, for any
  derivation~$D$ of $v$ in $H_q$
	$ w(D) = \bigotimes\limits_{t' \in \: leaves(\tau_D)} \mathsf{prov}_R^q(t') $
        where $\tau_D$ is the proof tree corresponding to~$D$ in the
        bijection given by Lemma~\ref{lemma:mult}.
\end{lemmarep}

\begin{proof}
 By induction on the size of the derivation $D$: 
 \begin{itemize}
	 \item If $|D| = 1$ then, there exists a nullary edge $e \in E_q$ with $h(e) = v$ and $w(D) = f_v = \mathsf{prov}_R^q(r(\vec{x})) = \prod\limits_{t' \in \: leaves(\tau_D)} \mathsf{prov}_R^q(t')$.
	 \item If $|D| \geq 1$ then there exists $e \in E_q$ and $D$ is of the form $ \langle e, D_1 \cdots D_{|e|} \rangle$ with $D_i$ a derivation of $T_i(e)$ for $1 \leq i \leq |e|$. 
	 We have $w(D) = f_e(w(D_1), \dots ,w(D_{|e|}))$ and by definition of $f_e = \otimes$ and by IHP $w(D) = \bigotimes\limits_{t' \in \: leaves(\tau_D)} \mathsf{prov}_R^q(t')$.
 \end{itemize}
\end{proof}

Finally, we obtain:

\begin{theoremrep} \label{th:eq}
	Let $t$ be a tuple of a Datalog program $q$ with output predicate
        $G$ and $H_q$ its hypergraph representation, then
        $\mathsf{prov}_G^q (t) = \delta_{H_{q}}(G(t))$.
\end{theoremrep}

\begin{proof}
  \[\begin{array}{cclcl}
	\delta_{H_{q}}(t)	&	= & \bigoplus\limits_{D \in
        \mathcal{D}_{H_{q}}(G(t))} w(D) &  & \mbox{and by Lemma~\ref{lemma:sum},} \\
					&	= &\bigoplus\limits_{D
                                        \in \mathcal{D}_{H_{q}}(G(t))}  \left(\bigotimes\limits_{t' \in \: leaves(\tau_D)} \mathsf{prov}_R^q(t') \right)& & \mbox{and by Lemma~\ref{lemma:mult},} \\ 
					&	= &\bigoplus\limits_{\tau \: \mathit{yields} \: t} \left( \bigotimes\limits_{t' \in \: leaves(\tau)} \mathsf{prov}_R^q(t') \right) & = & \mathsf{prov}_T^q (t)
	\end{array}
  \]
\end{proof}

\section{Best-first method}
\label{sec:algorithm}
	
Knuth~\cite{knuth_generalization_1977} generalized the Dijkstra algorithm
to what he calls the \emph{grammar problem} (i.e., finding the \textit{best-weight derivation} from a given non-terminal, 
where each terminal has a specific weight and each rule comes with an associated weight function). 
This has been identified as corresponding to the search problem in a monotonic superior hypergraph -- for each $e \in H$, $f_e$ is monotone and superior in each argument (see Table 3 in~\cite{huang_advanced_2008}). 
We showed in Lemma~\ref{lem:superiority} that
superiority corresponds to $0$-closedness in semirings with natural
orders.
The definition of the grammar problem assumes a total order on weights as the weights are real numbers. 
In the special case where all hyperedges are of arity~1 (and all weight
functions bound to~$\otimes$), we obtain the classical notion of semiring-based provenance for graph databases~\cite{ramusat:hal-01850510}. 
Thus, Knuth's algorithm can be seen as a generalization
to hypergraphs (and therefore, by the results of the previous section, to Datalog provenance computation) of the modified Dijkstra algorithm we presented in~\cite{ramusat:hal-01850510},
working on \emph{$0$-closed totally-ordered semirings}, which are
generalizations of the tropical semiring.



\begin{toappendix}
We present as Algorithm~\ref{alg:knuth} the Best-first method, the
  reformulation of the Knuth's algorithm in terms of semiring-based
  Datalog provenance, the basis of the optimizations introduced in the
  main text.

\begin{algorithm}
  \small
		\caption{Naïve version of Best-first method for Datalog provenance}
                \label{alg:knuth}
		\begin{algorithmic}[1]
			\Require{Datalog query $q$, $\mathsf{EDB}$ $D$ with
                        provenance indications over a $0$-closed totally-ordered semiring $S$.}
			\Ensure{full Datalog provenance for the $\mathsf{IDB}$ of
                        $q$ over $D$.}
				\State $I \leftarrow \emptyset$ 
				\State Let $\nu$ be the function that
                                maps all facts of $D$ to their 
                                annotation in~$S$ and all potential facts of the
                                intensional schema of $q$ to $\bar{0}$
				\Repeat
					\For{each intensional fact $r(\vec{x}) \notin I$ }\label{while-knuth}
						\State 
						$\nu(r(\vec{x})) \quad
                                                \oplus \! \! =
                                                \bigotimes\limits_{1 \leq
                                                i \leq
                                                n}\nu(r_i(\vec{x_i}))$ for each instantiation of a rule 
						
						\quad
                                                \quad$r(\vec{x})~\leftarrow~r_1(\vec{x_1}),
                                                \dots , r_n(\vec{x_n})$
                                                with $r_i(\vec{x_i}) \in
                                                D \cup I$ 
					\EndFor
                                \State{\label{line:extract-minimal} Add to $I$ the tuple
                                $r(\vec{x_{\min}})$
                                such that $\nu(r(\vec{x_{\min}}))$ is $\leq$-minimal
                                among all potential intensional facts
                                $r(\vec{x})$ not in $I$}
                                \Until $\nu(r(\vec{x_{\min}})) = \bar{0}$
                                or $I$ contains all potential intensional
                                facts
                        \State \Return $\nu_{|I}$
		\end{algorithmic}
\end{algorithm}
\end{toappendix}

\paragraph*{Optimized version of Best-first method} 

In the original paper of Knuth~\cite{knuth_generalization_1977}, 
the question of efficient construction of the set of candidate facts for the extraction of the minimal-valued fact is not dealt with. 
A lot of redundant work may be carried out if the implementation is not carefully designed.  

In the following, we show how to obtain a ready-to-be-implemented version incorporating ideas from the \textit{semi-naïve} evaluation of Datalog programs. 
Semi-naïve evaluation of Datalog, as described in~\cite[Chapter~13]{DBLP:books/aw/AbiteboulHV95} introduces a number of ideas aiming at improving the efficiency of the \textit{naïve} Datalog evaluation method; 
we show how to leverage these tricks in our setting.


The \textit{naïve} evaluation of a Datalog program $q$ processes iteratively, applying at each step the \textit{consequence operator}~$\mathrm{T}_q$. 
Many redundant derivations are computed, leading to practical inefficiency. 
The \emph{semi-naïve} evaluation addresses this problem by considering only facts derived using at least one new fact found at the previous step. 
Note, however, whereas many new facts can be found at one step of the semi-naïve evaluation, only one is to be added by the Best-first method to respect the $\leq$-minimality ordering of added facts.

\begin{algorithm}
  \small
		\caption{Basic semi-naïve version of Best-first method for Datalog provenance}
                \label{alg:knuth2}
		\begin{algorithmic}[1]
			\Require{Datalog query $q$, $\mathsf{EDB}$ $D$ with
			                        provenance indications over a $0$-closed totally-ordered semiring $S$.}
			\Ensure{full Datalog provenance for the $\mathsf{IDB}$ of $q$.}
				\Function{Relax}{$r_0(\vec{x_0})$, $S$}
					\For{each instantiation of a rule $r(\vec{x})~\leftarrow~r_1(\vec{x_1}), \dots, r_m(\vec{x_m}), \cdots , r_n(\vec{x_n})$ where $r_i(\vec{x_i}) \in D \cup S \cup \{ r_0(\vec{x_0}) \}$, $1\leq i < m$, $r_m(\vec{x_m}) = r_0(\vec{x_0})$ and $r_i(\vec{x_i}) \in D \cup S$, $m < i \leq n$}
						\State $\nu(r(\vec{x})) \quad \oplus \! \! = \bigotimes\limits_{1 \leq i \leq n}r_i(\vec{x_i})$ 			
					\EndFor
				\EndFunction
				\State

				\State{$I \leftarrow \emptyset$}
				\State Let $\nu$ be the function that
				                                maps all facts of $D$ to their 
				                                annotation in~$S$ and all potential facts of the
				                                intensional schema of $q$ to $\bar{0}$
				\For{each intensional atom $r(\vec{x}) \,  \notin \! I$ }\label{while-knuth2}
					\State
					$\nu(r(\vec{x})) \quad \oplus \! \! = \bigotimes\limits_{1 \leq i \leq n}r_i(\vec{x_i})$ for each instantiation of a rule 
					
					\quad $r(\vec{x})~\leftarrow~r_1(\vec{x_1}), \dots , r_n(\vec{x_n})$ with $r_i(\vec{x_i}) \in D$ 
				\EndFor
				\While {$ \mbox{min}_{\nu \setminus I } \: r(\vec{x}) \neq \bar{0}$}
					\State{Add such minimal $r(\vec{x})$ to $I$} \label{extract}
					\State \Call{Relax}{$r(\vec{x})$, $I \setminus r(\vec{x})$}
				\EndWhile
			\State \Return{$\nu$}
		\end{algorithmic}
\end{algorithm}

This algorithm starts by initializing the priority queue with $\mathsf{IDB}$ facts that are derivable from $\mathsf{EDB}$ facts. 
Then, at each step, the minimum valued-fact is added, and only
derivations using this new fact are computed to update the value of the
facts in $I$. This algorithm stops whether: 1. the maximal value is reached for a candidate fact, 2. the list is empty - the minimal value of the list is by default the maximal value of the semiring.

\begin{theorem}
	Algorithm~\ref{alg:knuth2} computes the full Datalog provenance for $0$-closed totally-ordered semirings.
\end{theorem}

\begin{proof}
	We show the algorithm verifies the following invariant: whenever a tuple is added to~$I$ in Line~\ref{extract}, it has optimal value. This implies that $I$ is populated in increasing order: each new derivation computed in the $\Call{Relax}$ procedure only updates the priority queues with values greater than the value of the tuple relaxed (by superiority of $\otimes$).
	
	Suppose by contradiction that some output tuples are not correctly labeled and take such a minimal tuple $\nu = r(\vec{x})$. At the moment where $\nu$ is extracted with value $n$ let us consider an optimal derivation path of $\nu$ that leads to the optimum value $opt < n$. By superiority each tuple occurring in the tail of the rule has value less than $opt$. Thus a tuple occurring in the tail is either wrong-valued or not present in $I$ at the moment where $\nu$ is found. In both cases and because tuples are added to $I$ in increasing order we obtain a new minimal tuple incorrectly labeled by the algorithm, contradicting the hypothesis.
\end{proof}



The structure of the Datalog program can be analysed to provide clues
about the predicates to focus on.
Following~\cite{DBLP:books/aw/AbiteboulHV95}, we introduce the notion of
\textit{precedence graph} $G_P$ of a Datalog program $P$. The nodes are
the $\mathsf{IDB}$ predicates and the edges are pairs of $\mathsf{IDB}$ predicates $(R,R')$ where $R'$ occurs at the head of a rule of $P$ with $R$ belonging to the tail.
$P$ is a \textit{recursive} program if $G_P$ has a directed cycle. 
Two predicates $R$ and $R'$ are mutually recursive if $R = R'$ or $R$ and
$R'$ participate in the same cycle of $G_P$. This defines equivalence
classes.
Putting it together, we obtain as Algorithm~\ref{alg:knuth3} out final
algorithm.

\begin{algorithm}[t]
		\caption{Semi-naïve version of Best-first method for Datalog provenance}
                \label{alg:knuth3}
		\begin{algorithmic}[1]
			\Require{$q$ a Datalog query with provenance indication over a $0$-closed totally-ordered semiring $S$.}
			\Ensure{full Datalog provenance for the $\mathsf{IDB}$ of $q$.}
				\State Compute the equivalence classes of $q$
				\For {each equivalence class in a topological order}
					\State Apply Algorithm~\ref{alg:knuth2} over $\mathsf{IDB}$ predicates in the equivalence class 
					\State considering previous equivalence classes as $\mathsf{EDB}$ predicates
				\EndFor
				\State \Return{$\nu$}
		\end{algorithmic}
\end{algorithm}


\paragraph*{Generalization to distributive lattices}

In~\cite{ramusat:hal-03140067}, we outlined a new algorithm, based on the Dijkstra algorithm and solving the single-source provenance in graph databases with provenance indications over $0$-closed multiplicatively idempotent semirings (equivalents to distributive lattices).
This new method stems from a tentative to bridge the strong complexity gap for computing the provenance in the case of a semiring not $0$-closed and totally ordered.
A similar gap also appears when we consider provenance over Datalog queries (see Section~\ref{sec:related-work}). Thus, we show how to apply this method for computing provenance for Datalog queries over distributive lattices.

We provide a brief review of the key ideas presented
in~\cite{ramusat:hal-03140067}. Any element of a distributive lattice is
decomposable into a product of \textit{join-irreducible} elements of the
lattice, and there exists an embedding of the distributive lattice into a
chain decomposition of its join-irreducible elements. This ensures we can
1) work on a totally ordered environment and apply algorithms that
require total ordering over the elements, 2) independently compute
partial provenance annotations for each dimension to form the final provenance annotation.
Given $m$ the number of dimensions in the decomposition, our solution
(described in Algorithm~\ref{alg:ext-knuth})
performs $m$ launches of the Best-first method and thus, roughly has a cost increased by a factor $m$.

\begin{figure}
\begin{algorithm}[H]
  \caption{Generalized Best-first method for Datalog provenance}
  \small
                \label{alg:ext-knuth}
		\begin{algorithmic}[1]
			\Require{$q$ a Datalog query with provenance indication over a $0$-closed multiplicatively idempotent semiring $S$.}
			\Ensure{full Datalog provenance for the $\mathsf{IDB}$ of $q$.}
			\For{each $\mathsf{EDB}$ fact $R(\vec{x})$}
				\State{\textsc{Decompose}($R(\vec{x})$)}
			\EndFor
			\For{each dimension $i$ }
				\State{$\nu_i \leftarrow \textsc{Best-first}(q,i)$}
			\EndFor
			\State \Return{\textsc{Recompose}($\nu_1, \ldots, \nu_n)$}
		\end{algorithmic}
\end{algorithm}

\begin{algorithm}[H]
  \caption{Input Datalog program computing the transitive closure
  (\textsc{Soufflé} syntax)}
\small
                \label{alg:tr-dl}
		\begin{algorithmic}[1]
			\State{\textbf{.decl} edge(s:number, t:number[, @prov:semiring value])}
			\State{\textbf{.input} edge}

			\State{\textbf{.decl} path(s:number, t:number[, @prov:semiring value])}
			\State{\textbf{.output} path}
			
			\State{path(x, y) :- edge(x, y).}
			\State{path(x, y) :- path(x, z), edge(z, y).}
		\end{algorithmic}
\end{algorithm}

\begin{algorithm}[H]
  \caption{Corresponding \textsc{Soufflé} \textsf{RAM} program for Algorithm~\ref{alg:tr-dl}}
\small
                \label{alg:ram}
		\begin{algorithmic}[1]
			\If{$\neg ($edge $= \emptyset)$}
				\For{t0 \textbf{in} edge}
					\State{\textbf{add} (t0.0, t0.1) \textbf{in} path}
					\State{\textbf{add} (t0.0, t0.1) \textbf{in} $\delta$path}
				\EndFor
			\EndIf
			\Loop
				\If{$\neg (\delta$path = $\emptyset) \wedge \neg ($edge $= \emptyset)$}
					\For{t0 \textbf{in} $\delta$path}
						\For{t1 \textbf{in} edge \textbf{on index} t1.0 = t0.1}
							\If{$\neg ($t0.0, t0.1) $\in$ path}
								\State{\textbf{add} (t0.0, t0.1) \textbf{in} path'}
							\EndIf
						\EndFor
					\EndFor
				\EndIf
				\If{path' = $\emptyset$}
					\State{\textbf{exit}}
				\EndIf
				\For{t0 \textbf{in} path'}
					\State{\textbf{add} (t0.0, t0.1) \textbf{in} path}
				\EndFor
				\State{\textbf{swap} $\delta$path \textbf{with} path'}
				\State{\textbf{clear} path}
			\EndLoop
		\end{algorithmic}
\end{algorithm}

\begin{algorithm}[H]
  \caption{Modification of \textsf{RAM} program of
  Algorithm~\ref{alg:ram} to implement Best-first strategy}
\small
                \label{alg:ram-prov}
		\begin{algorithmic}[1]
			\If{$\neg ($edge $= \emptyset)$}
				\For{t0 \textbf{in} edge}
					\State{\textbf{update} (t0.0, t0.1, t0.prov) \textbf{in} path}
					
				\EndFor
				\For{t0 \textbf{in} path}
					\State{\textbf{add} (t0.0, t0.1, t0.prov) \textbf{in} $\delta$path}
				\EndFor
			\EndIf
			\Loop
				\If{$\neg (\delta$path = $\emptyset) \wedge \neg ($edge $= \emptyset)$}
					\For{t0 \textbf{in} $\delta$path}
						\For{t1 \textbf{in} edge \textbf{on index} t1.0 = t0.1}
							\If{$\neg ($t0.0, t1.1, $\bot$) $\in$ path}
								\State{\textbf{update} (t0.0, t0.1, t0.prov $\otimes$ t1.prov) \textbf{in} \textsf{pq}}
							\EndIf
						\EndFor
					\EndFor
				\EndIf
				\State{\textbf{clear} $\delta$path}
				\If{\textsf{pq} is empty}
					\State{\textbf{exit}}
				\EndIf
				\State{\textbf{add} \textsf{pq}.top() \textbf{in} \textsf{pq}.top().relation \textbf{and in} \textsf{pq}.top().$\delta$relation}
			\EndLoop
		\end{algorithmic}
\end{algorithm}
\end{figure}

\section{Implementation and experiments}
\label{sec:implementation}

In numerous application domains, Datalog is used as a \emph{domain specific language} (DSL) to express logical specifications for static program analysis. 
A formal specification, written as a \textit{declarative} Datalog program is usually translated into an efficient \textit{imperative} implementation by a \textit{synthesizer}. 
This process simplifies the development of program analysis compared to hand-crafted solutions (highly optimized \textsc{C++} applications specialized in enforcing a fixed set of specifications).
\textsc{Soufflé}~\cite{10.1007/978-3-319-41540-6_23,10.1145/2892208.2892226} has been introduced to provide efficient synthesis of Datalog specifications to executable \textsc{C++} programs, competing with state-of-the-art handcrafted code for program analysis. 
The inner workings of \textsc{Soufflé} were of interest to our work; the algorithm implementations are similar to the evaluation strategy followed by the Best-first method we introduced here. We present a brief overview of the architecture of \textsc{Soufflé} and discuss how we extended it. 

\paragraph*{Architecture and implementation}

Following what is described in~\cite{10.1007/978-3-319-41540-6_23}, an input datalog
program~$q$ goes through a staged specialization hierarchy. 
After parsing, the first stage of \textsc{Soufflé} specializes the
semi-naïve evaluation strategy applied to $q$, yielding a relational algebra machine program (\textsf{RAM}). 
Such a program consists in basic relational algebra operations enriched with I/O operators and fix-point computations. 
As a final step, the \textsf{RAM} program is finally either interpreted or compiled into an executable. For this work, we have only used the interpreter.
Our code was inserted in two different stages of \textsc{Soufflé}:
a new \textit{translation strategy} from the parsed program
          to the \textsf{RAM} program,
a \textit{priority queue}, replacing the code in charge of adding at run-time the tuples to the relations.

We showcase the result of our translation strategy in
Algorithms~\ref{alg:tr-dl}, \ref{alg:ram}, and \ref{alg:ram-prov} for a
Datalog query computing the transitive closure of a graph; this program
is given in Algorithm~\ref{alg:tr-dl} in its \textsc{Soufflé} syntax.
Algorithm~\ref{alg:ram} presents the corresponding \textsc{Soufflé}
\textsf{RAM} program resulting from applying the semi-naïve evaluation
strategy and Algorithm~\ref{alg:ram-prov} our modification to the
\textsf{RAM} program to provide provenance annotation via the Best-first
strategy and use the priority queue \textsf{pq} for provenance
computation. The $\bot$ notation corresponds to a wildcard.
Importantly, modifying directly at the \textsf{RAM} level of \textsc{Soufflé} allows us to benefit of all implemented optimizations.

\begin{figure}[p]
	\centering
	\begin{tikzpicture}[scale=0.7]
	\begin{axis}[
	ybar, axis on top,
	height=8cm, width=12.5cm,
	bar width=0.4cm,
	ymajorgrids, tick align=inside,
	major grid style={draw=white},
	ymax=200000.0,
	x=2.6cm,
	ybar=1pt,
	axis x line*=bottom,
	axis y line*=right,       				
	y axis line style={opacity=0},
	ymode=log,
	log basis y={10},
	log origin=infty,
	tickwidth=4pt,
	enlarge x limits=0.15,
	legend style={
		at={(0.5,-0.1)},
		anchor=north,
		legend columns=-1,
		/tikz/every even column/.append style={column sep=0.5cm}
	},
	ylabel={time (s)},
	ylabel style={at={(-0.04,0.5)}},
	symbolic x coords={
		rome99,powergrid-rand,yeast-rand,stif-rand},
              xticklabels={
                {\textsc{Rome99}},{\textsc{USPowerGrid}},{\textsc{Yeast}},{\textsc{Stif}}},
	xtick=data,
	point meta=rawy,
	log ticks with fixed point,
	]
	\addplot [draw=none, fill=blue!30] coordinates {
		(rome99, 14.2149)
		(powergrid-rand, 0.0789086)
		(yeast-rand, 0.576678)
		(stif-rand, 491.538) };

	\addplot [draw=none, fill=green!45] coordinates {
		(rome99, 68.27066000000000000000)
      	(powergrid-rand,415.51890000000000000000) 
      	(yeast-rand,1397.18300000000000000000)
      	(stif-rand,172800) };
	\addplot [draw=none, fill=brown!55] coordinates {
		(rome99, 30.62224000000000000000)
		(powergrid-rand,11.86632000000000000000) 
		(yeast-rand,19.68208000000000000000)
		(stif-rand,1006.81200000000000000000) };
	\addplot [draw=none,fill=red!40] coordinates {
		(rome99, 52.1789)
		(powergrid-rand, 0.20201) 
		(yeast-rand, 1.8845)
		(stif-rand, 2081.38) };
	
	\legend{
			\algo{Soufflé},\algo{NodeElimination}-Id, \algo{NodeElimination}-Degree, \algo{Soufflé-prov} }
	\end{axis}
	\end{tikzpicture}
	\caption{Comparison between algorithms for all-pairs shortest-distances (Tropical). Values
        greater than 100\,000~s are timeouts.}
	\label{fig:comp}

        \vspace{-.5cm}
	\begin{tikzpicture}[scale=0.7]
	\begin{axis}[
	ybar, axis on top,
	height=8cm, width=12.5cm,
	bar width=0.4cm,
	ymajorgrids, tick align=inside,
	major grid style={draw=white},
	ymax=2.0,
	x=2.6cm,
	ybar=1pt,
	axis x line*=bottom,
	axis y line*=right,       				
	y axis line style={opacity=0},
	ymode=log,
	log basis y={10},
	log origin=infty,
	tickwidth=4pt,
	enlarge x limits=0.17,
	legend style={
		at={(0.5,-0.1)},
		anchor=north,
		legend columns=-1,
		/tikz/every even column/.append style={column sep=0.5cm}
	},
	ylabel={time (s)},
	ylabel style={at={(-0.04,0.5)}},
	symbolic x coords={
		rome99,powergrid-rand,yeast-rand,stif-rand},
              xticklabels={
                {\textsc{Rome99}},{\textsc{USPowerGrid}},{\textsc{Yeast}},{\textsc{Stif}}},
	xtick=data,
	point meta=rawy,
	log ticks with fixed point,
	]
	\addplot [draw=none, fill=yellow!55] coordinates { 
		(rome99, .00232036400000000000)
		(powergrid-rand,.00187317200000000000)
		(yeast-rand,.00264467300000000000)
		(stif-rand,.02814748000000000000) 
	};
	\addplot [draw=none, fill=blue!30] coordinates { 
		(rome99, 0.019)
		(powergrid-rand, 0.0105)
		(yeast-rand, 0.012)
		(stif-rand, 0.066) 
	};
	\addplot [draw=none, fill=green!45] coordinates { 
		(rome99, 0.00345099600000000000)
		(powergrid-rand,.00750463400000000000)
		(yeast-rand,.00389555900000000000)
		(stif-rand,.02285673000000000000)
    };
	\addplot [draw=none, fill=brown!55] coordinates { 
		(rome99, .05655670000000000000)
		(powergrid-rand,.19224890000000000000)
		(yeast-rand,.35900520000000000000)
		(stif-rand,.72772100000000000000)
	};
	\addplot [draw=none,fill=red!40] coordinates { 
		(rome99, 0.031)
		(powergrid-rand, 0.0128) 
		(yeast-rand, 0.016)
		(stif-rand, 0.17) 
	};
	
	\legend{\algo{BFS},
			\algo{Soufflé}, \algo{Dijkstra}, \algo{Mohri}, \algo{Soufflé-prov} }
	\end{axis}
	\end{tikzpicture}
	\caption{Comparison between algorithms for single-source shortest-distances (Tropical).}
	\label{fig:comp-ss}
	\begin{tikzpicture}[scale=0.7]
	  \begin{axis}[
	  			grid = major, 
	  			grid style = {dashed, gray!30}, 
	  			xlabel = {Output DB size},
	  			x label style={at={(current axis.right of origin)},anchor=north, below=5.5mm},
				ylabel = {time (s)},
				ylabel style={at={(-0.06,0.5)}},
				axis y line*=right,
				xmin = 0,   
				xmax = 3100,  
				ymin = 0,   
				ymax = 4, 
				/pgfplots/xtick = {0,500,...,3500},
				/pgfplots/ytick = {0,1,...,4},
				xticklabels={0.0~M, 0.5~M, 1.0~M, 1.5~M, 2.0~M, 2.5~M, 3.0~M},
				width=12cm, height=5cm,
				legend style={
					at={(0,1)},
					anchor=north west},
                                x label style={at={(axis description cs:0.5,0.05)},anchor=north},
	  			]
                \addplot[mark=*,red] table {souffle-prov-iris.dat};
                \addplot[mark=square*,blue] table {souffle-iris.dat};
   		\legend{{\algo{Soufflé-prov}}, {\algo{Soufflé}}}
	  \end{axis}
	\end{tikzpicture}
	\parbox[b]{2cm}{ \scriptsize 
					\begin{alignat*}{3}
					\textsf{ra(v, w, x, y, z)}	\; \;	&	\textsf{:-}		&& \; \textsf{p(v), p(w), p(x), p(y), p(z).}	\\
					\textsf{rb(v, w, x, y, z)}	\; \;	&	\textsf{:-}		&& \; \textsf{p(v), p(w), p(x), p(y), p(z).}	\\
					\textsf{q(v)}				\; \;	&	\textsf{:-}		&& \; \textsf{r(v, w, x, y, z).}				\\
					\textsf{r(v, w, x, y, z)}	\; \; 	&	\textsf{:-}		&& \; \textsf{ra(v, w, x, y, z),}				\\
														&					&& \; \textsf{rb(v, w, x, y, z).}				\\
					\textsf{q(z)}				\; \;	&	\textsf{:-}		&& \; \textsf{r(v, w, x, y, z).}
					\end{alignat*}
					}
	\caption{Comparison between \algo{Soufflé} and \textsc{Soufflé-prov} (IRIS)}
	\label{fig:comp-souffle-iris}

	\begin{tikzpicture}[scale=0.7]
	  \begin{axis}[
	  			grid = major, 
	  			grid style = {dashed, gray!30}, 
	  			xlabel = {Output DB size},
	  			x label style={at={(current axis.right of origin)},anchor=north, below=5.5mm},
				ylabel = {time (s)},
				ylabel style={at={(-0.06,0.5)}},
				axis y line*=right,
				xmin = 0,   
				xmax = 30,  
				ymode=log,
        			log basis y={10},
        			log origin=infty,
				/pgfplots/xtick = {0,5,...,30},
				xticklabels={0~K, 5~K, 10~K, 15~K, 20~K, 25~K, 30~K},
				width=12cm, height=5cm,
				legend style={
					at={(0,1)},
					anchor=north west},
                                x label style={at={(axis description cs:0.5,0.05)},anchor=north},
	  			]
                \addplot[mark=*,red] table {souffle-prov-amie.dat};
                \addplot[mark=square*,blue] table {souffle-amie.dat};
   		\legend{{\algo{Soufflé-prov}}, {\algo{Soufflé}}}
	  \end{axis}
	\end{tikzpicture}
	\parbox[b]{2cm}{ \scriptsize 
					\begin{alignat*}{3}
					\textsf{hc(x,y)}	 	\; \;	&	\textsf{:-}	&& \; \textsf{hasChild(x, y).}	\\
					\textsf{imt(x,y)}	\; \;	&	\textsf{:-}	&& \; \textsf{isMarriedTo(x, y).}	\\
					& && \\
					\textsf{hc(x,y)}	 	\; \;	&	\textsf{:-}	&& \; \textsf{imt(z,x), imt(z, y).}	\\
					\textsf{hc(x,y)}	 	\; \;	&	\textsf{:-}	&& \; \textsf{imt(x,z), imt(z, y).}	\\
					\textsf{imt(x, y)}	\; \;	&	\textsf{:-}	&& \; \textsf{imt(y, x).}			\\
					\textsf{imt(x, y)}	\; \;	&	\textsf{:-}	&& \; \textsf{hc(x, z), hc(y, z).}
					\end{alignat*}
					}
	\caption{Comparison between \algo{Soufflé} and \textsc{Soufflé-prov} (AMIE)}
	\label{fig:comp-souffle-amie}
\end{figure}

\paragraph*{Experiments}

Our implementation was tested on an Intel Xeon E5-2650 computer with 176~GB of
RAM. 
The source code
is freely available on \textsc{Github}\footnote{\url{https://github.com/yannramusat/souffle-prov}}.


The initial motivation for this work stems from a key observation we outlined in the conclusion of~\cite{ramusat:hal-03140067}, where we pointed out the similarity between Datalog and the classes of semirings and their optimized provenance algorithms discussed in that work, focused on graph provenance algorithms. To translate this graph setting into Datalog, 
the graph structure has been encoded into an $\mathsf{EDB}$ with one binary predicate \emph{edge} encoding the edges, and with edge notations depending on the provenance semiring we chose. We run the \emph{transitive closure} Datalog program outlined in Algorithm~\ref{alg:tr-dl}. Full information on the graph datasets used can be found in~\cite{ramusat:hal-03140067}. 
We provide, in Figure~\ref{fig:comp}, a comparison between the best-first
method introduced here (\textsc{Soufflé-prov}), the plain
\textsc{Soufflé} without provenance computation, and a previous provenance-based algorithm
\cite{ramusat:hal-03140067} computing
all-pairs shortest-distances over graph databases (the
\textsc{NodeElimination} algorithm, with a choice of node to eliminate based
on its id or its degree), in the tropical semiring.
Similarly, in Figure~\ref{fig:comp-ss}, we compare with previous
solutions for single-source shortest-distances, in the same semiring, in
particular the adaptation of \textsc{Dijkstra} algorithm
of~\cite{ramusat:hal-03140067}, and, for comparison purposes, a
\textsc{Bfs} algorithm that does not compute provenance.
The main focus of this work was to provide an effective Datalog based solution for all-pairs provenance in graph databases. 
For the all-pairs problem, depending on the dataset, (see, e.g.,
\textsc{Yeast}), \textsc{Soufflé-Prov} is significantly faster than the previous best known algorithm, \textsc{NodeElimination}. 
Unsurprisingly, \textsc{BFS} and \textsc{Dijkstra} perform respectively better than \textsc{Soufflé} and \textsc{Soufflé-prov} in the single-source context. What favors both graph algorithms strongly is the fact that they reduce redundant computation: the algorithms abort whenever the target vertex has been reached. 
\textsc{Soufflé-prov} performs between $1$ and $2$ orders of magnitude
faster than \textsc{Mohri} \cite{Mohri:2002:SFA:639508.639512} -- an algorithm designed for single-source provenance on $k$-closed semirings. This fact highlights the potential of adapting the best-first method to also handle $k$-closed semirings.   


We now turn to evaluating the overhead induced by adding provenance
computations to Datalog via \textsc{Soufflé}. This appears fairly modest
in the experiments of Figures~\ref{fig:comp}, \ref{fig:comp-ss}. We further
consider two datasets that go beyond the setting of graph databases:
a non-recursive query over synthetic data (IRIS \cite{7113407}, obtained
from the authors of that paper) in
Figure~\ref{fig:comp-souffle-iris} and a program over a knowledge base
(AMIE
\cite{amie}, available
online\footnote{\url{https://bitbucket.org/amirgilad/selp/src/Journal/}}) populated with real-world data and automatic rules in Figure~\ref{fig:comp-souffle-amie}. The evaluation was made by varying the output database size.
In IRIS, the overhead induced is a modest constant factor of
approximately $4$ times the original cost. This is not the case for AMIE,
where the overhead tends to increase with the size of the answer. We
conjecture this is due to the lack of native support for tuple updates in \textsc{Soufflé} (in Datalog, the fixed-point operator is inflationary); this leads to inefficiency in updating provenance values associated to each tuple.

\section{Reverse translations and opportunities}
\label{sec:opportunities}

Our focus so far in this study was to express Datalog programs into hypergraphs. We now consider the opposite direction: given a weighted hypergraph $H$ with all its weight functions derived from a semiring $S$ via its $\otimes$ operator we aim to design an equivalent Datalog program which computes the provenance on the given semiring.
We propose two translations -- a straightforward one and one with a fixed
program, to benefit of the low data complexity of Datalog.

\begin{definition}[Weighted hypergraph to Datalog program]
	Given an hypergraph $H = \langle V,E \rangle$ of maximal hyperedge arity $n$, we define a Datalog query $q_H$ with extensional predicates $E_n$ of arity $n+1$ for $0 \leq i \leq n$ and an unary intensional predicate $R$. 
	For each $1 \leq i \leq n$, we add a single rule of the form: 
  \(
		R(x)		 \leftarrow 		E_{i+1}(x, x_1, \dots, x_n),
							R(x_1), \dots, R(x_n)
                                              \),
	 degenerating into $R(x) \leftarrow E_1(x)$ for $i = 0$.
	For each $e \in E$ of arity $n \geq 1$, we populate $E_{n+1}$
        with $E_{n+1}(h(e), {\mathrm{T}(e)})$ having provenance value $\bar{1}$. 
	Source vertices (i.e., heads of nullary edges) may have an initial constant value $s \in S$. For each $e \in E$ of arity $0$ we tag $E_1(h(e))$ with the provenance value $f_e$.
\end{definition}

\begin{toappendix}
Let us start with a lemma showing the one-to-one correspondence between derivations in the hypergraph and proofs in Datalog. 

\begin{lemmarep} \label{lemma:mult_rev}
	Given a vertex $v \in V$ of the hypergraph, $\mathcal{D}_{H}(v) \simeq \{\tau \mid \tau \mbox{ yields } v\}$.
\end{lemmarep}

\begin{proof}
  By definition of $q_H$ each rule corresponds to a unique hyperedge. Then, we can inductively construct for a given derivation $D$ its associated (unique) Datalog proof $\tau_D$:
	\begin{itemize}
		\item If $|D| = 1$, then $v$ is a source vertex (head of a nullary hyperedge). Thus, the corresponding proof is the single instantiation of the rule $R(v) \leftarrow E_1(v)$.
		\item If $|D| \geq 1$, then it exists $e \in \mathit{BS}(v)$ where $|e| \geq 0$ and $D_i$ is a derivation of $T_i(e)$ for $1 \leq i \leq |e|$, then $D =  \langle e, D_1 \cdots D_{|e|} \rangle$ is a derivation of $v$. 
		By definition, this edge corresponds to a single rule $R(v) \leftarrow E_{i+1}(v, T_1(e), \dots, T_{|e|}(e)), R(T_1(e)), \dots, R(T_{|e|}(e))$.
		By IHP, for $1 \leq i \leq |e|$, $\tau_{D_i}$ is the corresponding proof of the derivation $D_i$. Then by composition we obtain $\tau_D$ the proof for $D$.
	\end{itemize}
\end{proof}

We then show the weight of each derivation of a tuple is equal to the corresponding proof weight in Datalog.

\begin{lemmarep} \label{lemma:sum_rev}
$ w(D) = \bigotimes\limits_{t' \in \: \text{leaves}(\tau_D)} \mathsf{prov}_{E_1, \dots, E_{n+1}}^{q_H}(t') $
\end{lemmarep}

\begin{proof}
	By induction on the size of the derivation $D$: 
	\begin{itemize}
		\item If $|D| = 1$ then, there exists a nullary edge $e \in E$ with $h(e) = v$ and $w(D) = f_e = \mathsf{prov}_{E_1}^{q_H}(e) = \bigotimes\limits_{t' \in \: \mathit{leaves}(\tau_D)} \mathsf{prov}_{E_1, \dots, E_{n+1}}^{q_H}(t')$.
		\item If $|D| \geq 1$ then it exists $e \in E$ and $D$ is of the form $ \langle e, D_1 \cdots D_{|e|} \rangle$ with $D_i$ a derivation of $T_i(e)$ for $1 \leq i \leq |e|$. We have $w(D) = f_e(w(D_1), \dots ,w(D_{|e|}))$ and by definition of $f_e = 		\otimes$ and by IHP $w(D) = \bigotimes\limits_{t' \in \: \mathit{leaves}(\tau_D)} \mathsf{prov}_{E_1, \dots, E_{n+1}}^{q_H}(t')$.
	\end{itemize}
\end{proof}
\end{toappendix}

We then have:

\begin{theoremrep} \label{th:eq_rev}
	Let $v$ be a vertex of a weighted hypergraph $H$ with all weight functions derived from the $\otimes$ operation of a semiring $S$ and $q_H$ its translation into a Datalog query, then $\delta_{H}(v) = \mathsf{prov}_{R}^{q_H} (v)$.
\end{theoremrep}

\begin{proof}[Proof of Theorem~\ref{th:eq_rev}]
  \[ 	\begin{array}{cclcl}
	\delta_{H}(t)	&	= & \bigoplus\limits_{D \in \mathcal{D}_{H}(v)} w(D) &  & \mbox{and by Lemma~\ref{lemma:sum_rev},} \\
					&	= &\bigoplus\limits_{D \in \mathcal{D}_{H}(v)}  \left(\bigotimes\limits_{t' \in \: leaves(\tau_D)} \mathsf{prov}_{E_1, \dots, E_{n+1}}^{q_H}(t') \right)& & \mbox{and by Lemma~\ref{lemma:mult_rev},} \\ 
					&	= &\bigoplus\limits_{\tau \: \mathit{yields} \: t} \left( \bigotimes\limits_{t' \in \: leaves(\tau)} \mathsf{prov}_{E_1, \dots, E_{n+1}}^{q_H}(t') \right) & = & \mathsf{prov}_R^{q_H} (t)
	\end{array}
  \]
\end{proof}

Let us know consider the case of programs where we have a fixed schema and arity:

\begin{definition}[Weighted hypergraph to Datalog program with fixed schema]
	Given an hypergraph $H = \langle V,E \rangle$ let us consider the following Datalog program $q_{H_f}$ over unary predicates $R$ and $\mathit{Nullary}$, binary predicates $E$, $N$ and $\mathit{First}$ and ternary predicate $\mathit{Next}$:	
	\begin{align*}
			R(x)		&	\leftarrow		E(x, e), H(e) \\
			H(e)		&	\leftarrow		\mathit{First}(e, x), R(x), N(e,x) 
                        &N(e,x)	&	\leftarrow		\mathit{Next}(e, x, y), R(y), N(e, y) \\
					&	\leftarrow		\mathit{Nullary}(e) 
                        &		&	\leftarrow		\mathit{End}(e, x)
	\end{align*}
	
	For each $e \in E$ of arity $n \geq 1$, we populate $E$ with $E(h(e), e)$ with provenance value~$\bar{1}$. 
        Let $x_1, \dots, x_n$ be the elements of ${\mathrm{T}(e)})$, we populate $\textit{First}$ with $\textit{First}(e,x_1)$, $\mathit{Next}$ with $\mathit{Next}(e, x_i, y_{i+1})$ for $1 \leq i \leq n-1$ and $End$ with $\mathit{End}(e, x_n)$; all have provenance value~$\bar{1}$. 
	
	Source vertices (heads of nullary edges) may have an initial constant value $s \in S$, then for each $e \in E$ of arity $0$ we tag $E(h(e),e)$ with the provenance value $f_e$ and $\mathit{Nullary}(e)$ with provenance value $\bar{1}$.
	
\end{definition}

The reasoning follows a similar flow:

\begin{toappendix}
\begin{lemmarep} \label{lemma:mult_fixed_rev}
	Given a vertex $v \in V$ of the hypergraph, $\mathcal{D}_{H_f}(v) \simeq \{\tau \mid \tau \mbox{ yields } v\}$.
\end{lemmarep}

\begin{proof} 
	We inductively construct for a given derivation $D$ its associated (unique) Datalog proof $\tau_D$:
	\begin{itemize}
		\item If $|D| = 1$, then $v$ is a source vertex (head of a nullary hyperedge $e$). Thus, the corresponding proof is composed of $R(v) \leftarrow E(v, e), H(e)$ and $H(e) \leftarrow \mathit{Nullary}(e)$.
		\item If $|D| \geq 1$, then there exists $e \in \mathrm{BS}(v)$ where $|e| \geq 0$ and $D_i$ is a derivation of $T_i(e)$ for $1 \leq i \leq |e|$, then $D =  \langle e, D_1 \cdots D_{|e|} \rangle$ is a derivation of $v$. 
                  Let $v_1, \dots, v_n$ be the elements of
                  ${\mathrm{T}(e)})$, this hyperedge corresponds to the following proof tree: 
                \[
		\begin{array}{lll} 
			R(v)		&	\leftarrow	&	E(v, e), H(e) \\
			H(e)		&	\leftarrow	&	\mathit{First}(e, v_1), R(v_1), N(e, v_1) \\
			\multicolumn{3}{l}{For \: 1 \leq i \leq n-1: } \\
			N(e,v_i)	&	\leftarrow	&	\mathit{Next}(e, v_i, v_{i+1}), R(v_{i+1}), N(e, v_{i+1}) \\
			N(e,v_n)	&	\leftarrow	&	\mathit{End}(e, v_n)
		\end{array}
              \]
		By IHP, for $1 \leq i \leq |e|$, $\tau_{D_i}$ is the corresponding proof of the derivation $D_i$ (derivation of the fact $R(v_i)$). Then, by composition, we obtain $\tau_D$ the proof for $D$.
	\end{itemize}
\end{proof}

\begin{lemmarep} \label{lemma:sum_fixed_rev}
	$ w(D) = \bigotimes\limits_{t' \in \: leaves(\tau_D)} \mathsf{prov}_{E, \mathit{Nullary}, \mathit{First}, \mathit{Next}, \mathit{End}}^{q_{H_f}}(t') $
\end{lemmarep}

\begin{proof}
	By induction on the size of the derivation $D$: 
	\begin{itemize}
		\item If $|D| = 1$ then there exists a nullary edge $e \in E$ with $h(e) = v$ and its associated (unique) derivation is $R(v) \leftarrow E(v, e), H(e)$ and $H(e) \leftarrow \mathit{Nullary}(e)$. 
		Thus $w(D) = f_e = \mathsf{prov}_{E}^{q_{H_f}}(v,e) \: \otimes \: \mathsf{prov}_{\mathit{Nullary}}^{q_{H_f}}(e)$.
		\item If $|D| \geq 1$ then there exists $e \in E$ and $D$ is of the form $ \langle e, D_1 \cdots D_{|e|} \rangle$ with $D_i$ a derivation of $T_i(e)$ for $1 \leq i \leq |e|$.
		The derivation tree given in Lemma~\ref{lemma:mult_fixed_rev} corresponding to the hyperedge has provenance:
		\begin{equation*}
			\begin{split} 
	 			\mathsf{prov}_{E}^{q_{H_f}}(v,e) \: \otimes \: \mathsf{prov}_{\mathit{First}}^{q_{H_f}}(e) \: \otimes \: \mathsf{prov}_{R}^{q_{H_f}}(v_1) \: \otimes & \\
	 			\bigotimes\limits_{1 \leq i \leq n- 1} \left( \mathsf{prov}_{\mathit{Next}}^{q_{H_f}}(e,v_i, v_{i+1}) \otimes \mathsf{prov}_{R}^{q_{H_f}}(v_{i+1}) \right) \otimes \mathsf{prov}_{\mathit{End}}^{q_{H_f}}(e, v_n) 
	 		\end{split}
	 	\end{equation*}  
	 	By IHP, for $1 \leq i \leq |e|$, $w(D_i) =
                \mathsf{prov}_{R}^{q_{H_f}}(v_i)$, thus, we obtain \[w(D)
                = \bigotimes\limits_{t' \in \: leaves(\tau_D)}
                \mathsf{prov}_{E, \mathit{Nullary}, \mathit{First},
                \mathit{Next}, \mathit{End}}^{q_{H_f}}(t').\]
	\end{itemize}
\end{proof}
\end{toappendix}

\begin{theoremrep} \label{th:eq_fixed_rev}
	Let $v$ be a vertex of a weighted hypergraph $H$ with all weight functions derived from the $\otimes$ operation of a semiring $S$ and $q_{H_f}$ its translation into a Datalog query with fixed schema, 
	then $\delta_{H_f}(v) = \mathsf{prov}_{R}^{q_{H_f}} (v)$.
\end{theoremrep}

\begin{proof}[Proof of Theorem~\ref{th:eq_fixed_rev}]
	We note $\text{EDB} = \{E, \mathit{Nullary}, \mathit{First}, \mathit{Next}, \mathit{End}\}$.
        \[ 	\begin{array}{cclcl}
		\delta_{H_f}(t)	&	= & \bigoplus\limits_{D \in \mathcal{D}_{H_f}(v)} w(D) &  & \mbox{and by Lemma~\ref{lemma:sum_fixed_rev},} \\
						&	= &\bigoplus\limits_{D \in \mathcal{D}_{H_f}(v)}  \left(\bigotimes\limits_{t' \in \: \mathit{leaves}(\tau_D)} \mathsf{prov}_{\text{EDB}}^{q_{H_f}}(t') \right)& & \mbox{and by Lemma~\ref{lemma:mult_fixed_rev},} \\ 
						&	= &\bigoplus\limits_{\tau \: \mathit{yields} \: t} \left( \bigotimes\limits_{t' \in \: \mathit{leaves}(\tau)} \mathsf{prov}_{\text{EDB}}^{q_{H_f}}(t') \right) & = & \mathsf{prov}_R^{q_{H_f}} (t)
		\end{array}
        \]
\end{proof}

\paragraph*{Case study: AND/OR graphs} 
Section 4.3 of~\cite{huang_advanced_2008} showcases some formalisms in
which equivalent hypergraphs can be constructed in the dynamic
programming framework proposed in their work. One advantage is that we
can directly benefit from these translations in our Datalog provenance
framework. However, there are some issues with these translations, as we
now discuss.

One example formalism that can be easily translated to our setting are the AND/OR graphs, a special case of graphs -- e.g., used in scheduling tasks -- in which nodes can express two types of restrictions: AND restrictions (in which tasks have to executed sequentially) and OR restrictions (in which tasks can be overlapping). 
These two operations, not unlike the $\otimes$ and $\oplus$ operations on semirings, suggest our translation strategy is readily usable.
Indeed, by our translation, vertices of the hypergraph would represent OR vertices and the hyperedges correspond to the AND vertices. 
This translation does not add any other restriction to the shape of the AND/OR graph. 
AND and OR vertices can have multiples outgoing edges and the AND/OR graph is not necessarily bipartite. 
However, not all applications in which AND/OR graphs are useful can be used with semiring operations. For instance, scheduling using AND/OR graphs~\cite{fast-pathfinding-and-or, scheduling-and-or-prec-constraints} requires three operations (\texttt{min}, \texttt{max} and $+$) instead of the two allowed using semirings.

Another crucial issue, not obvious at first glance, is that when expressing these graphs into our Datalog framework the semantics can change. 
As explained in~\cite{scheduling-and-or-prec-constraints}, zero-edge cycles express mutual events (occurring in groups, with no causal relations): either they all occur together or none of them do. 
In a deductive system, where semantics is given by a least fixed-point, they are believed to not occur at all. 
The semantics of Datalog provenance or derivation trees cannot express this situation: in Datalog there would be no finite proof of a set of given facts.  However, adding all of these facts to the solution would not break the deductive rules; simply they do not belong to the least fixed-point.



\section{Related work}
\label{sec:related-work}

There are two major notions of provenance for information systems currently in use in the literature, each focusing on different usages. 
The first notion can be broadly categorized as ``informational'': the provenance encapsulates information about the deductive process leading to a specific result. 
Declarative debugging is one application of this concept. Analysis of provenance information in an interactive manner simplifies the user's burden of identifying which part of the rule specification is responsible for a \textit{faulty} derivation. 
\textsc{Soufflé} contains its own integrated debugging system, based on a provenance-based evaluation strategy~\cite{10.1145/3379446}.

In the current study we consider a ``computational'' notion of provenance, where operations (and queries) over provenance values are permitted. 
In this case, the underlying properties of the semiring are important for optimized algorithms.
For instance, \textit{absorptivity} is an ubiquitous property in the literature allowing efficient algorithms for provenance evaluation. This applies not only for Datalog programs, but also for e.g., regular path queries over (semiring)-annotated graph databases~\cite{ramusat:hal-03140067}.

With respect to Datalog provenance, it has been shown
in~\cite{deutch_circuits_2014} that, for a Datalog program having $n$
candidate $\mathsf{IDB}$ tuples, a circuit for representing Datalog
provenance in the semiring $\textbf{Sorp}(X)$ (the most general
absorptive semiring) only needs $n+1$ layers. For binary relations, e.g.,
representing the edge relation of a graph, this construction is at least
quadratic in the number of vertices, thus not practically applicable for
the graphs we analyzed in our experiments. Similarly,
in~\cite{esparza_solving_2011}, absorptive semirings (i.e., $0$-closed
semirings) have the property that derivation trees of size $\geq n$ are ``pumpable'' (they do not contribute to the final result). A concrete implementation~\cite{10.1007/978-3-319-08846-4_1} computes the provenance for commutative and idempotent semirings using $n$ Newton iterations.

Fairly recently, \cite{DBLP:journals/corr/abs-2105-14435} introduced POPS (\textit{Partially Ordered, Pre-Semiring}),  a structure decoupling the order on which the fixed-point is computed from the semiring structure. 
Complex and recursive computations over vectors, matrices, tensors are now expressible using this framework. The study also generalized the semi-naïve method from plain Datalog evaluation to idempotent semirings (aka dioids). 
In comparison, our method is restricted to semirings that are totally ordered (a subclass of distributive dioids\footnote{Distributive dioids are POPS structures over a distributive lattice being the natural order of the dioid.}), leveraging the invariant that once a fact is first labeled with a provenance value, we are certain it is the correct one.

In cases where keeping the \textit{full} provenance of a program (\emph{how}-provenance) is still prohibitively large, \cite{Deutch2018EfficientPT, 7113407} propose to select only a relevant subset of such trees using \textit{selection criteria} based on tree patterns and ranking over rules and facts occurring in the derivation. 
First, given a Datalog program $P$ and a pattern $q$, an \textit{offline} instrumentation is performed, leading to an \textit{instrumented program} $P_q$.
Then, given any database $D$, an efficient algorithm can be used to retrieve only the top-$k$ best derivation trees for $P_q(D)$. 
The top-$1$ algorithm of the study is closely related to our solution, but does not mention the use of a priority queue nor does it take into account the optimization provided by the semi-naïve evaluation strategy we describe in Section~\ref{sec:algorithm}. 

Our solution can be seen as a hybrid of the ideas introduced
in~\cite{DBLP:journals/corr/abs-2105-14435}
and~\cite{Deutch2018EfficientPT}. We generalize the semi-naïve evaluation
to a specific class of semirings in order to achieve a more
efficient algorithm, one that can be used in practical real-world scenarios.

\section{Conclusion}
\label{sec:conclusion}

In this work, we developed a novel method for Datalog provenance
computation based on the link between dynamic programming over
hypergraphs and the proof structure of provenance of Datalog programs. We
introduced Knuth's algorithm for computing the provenance, and optimized
it for practical use. We showed its feasibility by providing an
implementation on top of \textsc{Soufflé} and tested it on several graph databases and Datalog programs.

As a continuation of this work, we aim to establish a taxonomy of recent algorithms for provenance computation of Datalog programs (e.g.~\cite{DBLP:journals/corr/abs-2105-14435}), i.e., classify them by their data complexity and the class of semiring usable in each algorithm -- with a focus on $k$-closed semirings.
The objective is to provide a meaningful comparison with the already
established taxonomy for semiring provenance in graph
databases~\cite{ramusat:hal-03140067}.  Another direction of interest
would be to investigate a ``fix'' to the problems raised by Datalog provenance in AND/OR graphs (see Section~\ref{sec:opportunities}). Finally, in a more practical direction,  it might be feasible to implement
the semi-naïve evaluation for Datalog$^{\text{o}}$ over distributive dioids (Algorithm 1 in~\cite{DBLP:journals/corr/abs-2105-14435}) using \textsc{Soufflé}. Most of the inner workings can easily be adapted to do so, but extending the current data structures storing facts (B-Trie, B-Tree) to allow updates may be a challenge.

\clearpage
\bibliography{main}

\end{document}